\documentclass[journal,onecolumn]{IEEEtran}
\ifCLASSINFOpdf
\else
\fi

	\usepackage{amssymb}
	\usepackage{graphicx}%
	\usepackage{amsmath,amssymb,amsfonts}%
	\usepackage{amsthm}%
	\usepackage{mathrsfs}%
	\usepackage{xcolor}%
	\usepackage{physics}
	\usepackage{hyperref}
	\usepackage{multirow}%
	
	\usepackage{textcomp}%
	\usepackage{manyfoot}%
	\usepackage{booktabs}%
	\usepackage{algorithm}%
	\usepackage{algorithmicx}%
	\usepackage{tcolorbox}%
	\usepackage{algpseudocode}%
	\usepackage{listings}%

	\theoremstyle{thmstyleone}%
	\newtheorem{theorem}{Theorem}

	\theoremstyle{thmstyletwo}%

	\theoremstyle{thmstylethree}%
	\newtheorem{definition}{Definition}%

\begin{document}
		%
		\title{On the Optimality of a Quantum Key Distribution}
		%
		%
		%
		
		\author{Georgi~Bebrov
			
			\thanks{G. Bebrov was with the Department
				of Telecommunications, Technical University of Varna, Varna, 9010 BULGARIA e-mail: g.bebrov@tu-varna.bg.}
			}
		
		%
		%

	\markboth{Journal of \LaTeX\ Class Files}%
	{Shell \MakeLowercase{\textit{et al.}}: Bare Demo of IEEEtran.cls for IEEE Journals}
	%



	\maketitle
	
	\begin{abstract}
		Quantum key distribution (QKD) systems require optimal performance of both quantum and classical channels\textemdash utilizing as few as possible qubits and bits for establishing as many as possible key bits. Here we report a way to determine if a quantum key distribution model (or protocol) operates in an optimal behavior. This is accomplished by introducing a quantity, called \textit{optimality}, which is the maximum over the total efficiency of a QKD under any circumstances (any values of QKD parameters). The optimality definition is given for the asymptotic operation of a QKD system\textemdash when infinitely many quantum systems are transferred/used in a quantum key distribution protocol \textit{or} a quantum key distribution system is used infinitely many times. A way to attain the optimality is considered\textemdash implementation of a completely efficient QKD system (a combination of capacity-reaching quantum channel and a completely compressed classical channel) is presented. Optimal versions of BB84-QKD and twin-field QKD are introduced.
	\end{abstract}
	

	%
	\IEEEpeerreviewmaketitle
	
	\section{Introduction}\label{sec1}
	\IEEEPARstart{T}{he} process of \textit{quantum key distribution} is a technology that enables establishing cryptographic keys between two or more parties (participants). Its security relies upon the principles of quantum mechanics. Throughout the years several prominent QKD models have been developed \cite{Bennett1984,Ekert1991,Lo2005,Lo2012,Lucamarini2018,Curty2019,Wang2018}. As with any other system (technology), quantum key distribution should reach after its optimal operation (performance). The optimal performance is related to using efficiently both quantum and classical channels of a QKD system. As shown in the literature \cite{Lo2005,Curty2019}, it is possible for a QKD implementation (existing QKD protocol) to have an efficient quantum channel. This is achieved by exploiting biased bases\textemdash one of the preparation/measurement bases is characterized by a considerably greater probability of occurrence $p$ in comparison to the other $\text{1}-p$, \textit{\textit{i.e.}}, $p\rightarrow\text{1}$. In the asymptotic limit $p\rightarrow\text{1}$, the quantum channel of an implementable QKD approximately resembles an ideal quantum channel. On the other hand, no one in the literature contemplates an efficient classical channel (reduction in the use of a classical channel). In this connection, this article presents a way for a quantum key distribution to reach an optimal behavior. We develop a \textit{compression} (or \textit{coding}) \cite{Huffman1952} process to attain an efficient classical channel (Section \ref{compression}). Note that the compression process leads to an efficient classical channel in the asymptotic limit of $p\rightarrow\text{1}$ (completely biased bases are used) 
	and $n\rightarrow\infty$ ($n$\textemdash number of transferred quantum systems in a QKD protocol or number of uses of a QKD system). Optimality of a QKD is reached by combining both the concept of biased preparation/measurement bases and compression. There is also a lack of optimality evaluation in the literature.
	In this regard, the work has the following objectives
	\begin{enumerate}
		\item proposing a definition for optimality of a QKD
		\item proposing a method of determining the optimality of a QKD 
		\item proposing a compression appropriate to obtain an optimal QKD
		\item proposing optimal BB84-QKD in order to justify the optimality evaluation.
	\end{enumerate}
	Note that the proposed method could be used as a tool in the process of developing a QKD and in the process of comparing QKD models.\\
	\indent The paper is organized as follows. Section~\ref{bb84-steps} recalls in a resource-wise manner the BB84-QKD protocol that is used as a figure of merit for defining optimality of a QKD. Section \ref{opt-def} presents a definition for optimality of a QKD. Section \ref{opt-eval} presents optimality evaluation for a BB84-QKD. In this section, a comparison is presented between the total efficiencies of optimal BB84 and standard BB84 in order to highlight the application of the optimality proposed herein. Section \ref{compression} introduces a compression process (called \textit{channel squeezing}), which could be used to attain an efficient classical channel, correspondingly an optimal QKD protocol. Section \ref{protocols} proposes two optimal QKD protocols\textemdash one of BB84 type \cite{Bennett1984} and one of twin-field (TF) type \cite{Lucamarini2018,Curty2019}. Section \ref{summary} put forwards a summary.

	\section{BB84-QKD}\label{bb84-steps}
	For the purposes of this paper, we recall BB84-QKD in a resourse-wise perspective. The protocol consists of the following steps:\\
	

	\begin{tcolorbox}[colback=yellow!5,colframe=yellow!40!black,title=\textsc{Qubit transfer}]
		\textbf{\textit{q-1})} $N$ qubits are transferred from \textit{Alice} to \textit{Bob}. Each transferred qubit corresponds to a bit: $0$ or $1$ (the value depends on the state of the transferred qubit).\\
		\textbf{\textit{q-2})} $\Tilde{\eta}N$ qubits are received by \textit{Bob}, where $\Tilde{\eta}=\eta\eta_{\text{det}}$ denotes the total transmittance of the QKD link ($\eta=10^{-\alpha L/10}$--channel transmittance; $\eta_{\text{det}}$--detection efficiency). Note that $L$ is the length of QKD link (in kilometers) and $\alpha$ is the channel attenuation. Each received/detected qubit corresponds to a bit: $0$ or $1$ (the value depends on the state of the received qubit).\\
	\end{tcolorbox}
	
	\begin{tcolorbox}[colback=blue!5,colframe=blue!40!black,title=\textsc{Sifting}]
		\textbf{\textit{c-1})} $N$ bits are announced to acknowledge proper reception of qubits ($0\rightarrow$ not received/detected; $1\rightarrow$ received/detected). Note that \textit{Bob} is aware of the time windows at which qubits should arrive at his detection system. \\
		\textbf{\textit{c-2})} $(1-\sigma)\Tilde{\eta}N$ bits are broadcasted by \textit{Bob} to announce measurement bases used in detection step \textbf{\textit{q-2})} ($0\rightarrow$ rectilinear $Z$ basis; $1\rightarrow$ diagonal $X$ basis). The quantity $\sigma$ ($\sigma\in[0,1]$) denotes the \textit{compression} of the classical channel when a proper encoding is applied to the announced bits.\\
		\textbf{\textit{c-3})} $(1-\sigma)\Tilde{\eta}N$ bits are broadcasted by \textit{Alice} to acknowledge preparation-measurement bases match ($0\rightarrow$ match of preparation and measurement bases for a certain qubit received by \textit{Bob}; $1\rightarrow$ no match of preparation and measurement bases for a certain qubit received by \textit{Bob}).\\
		\textbf{\textit{c-4})} \textit{Alice} and \textit{Bob} discard all the qubits (bits transferred/received by qubits) for which there is no match between preparation and measurement bases. The variable $s$ (\textit{parameter of sifting}) denotes the percentage of preparation-measurement bases match ($s\in[0,1]$). As a result, only $s\Tilde{\eta}N$ bits remain at the end of the sifting procedure. Those bits compose the so-called \textit{sifted key}.
	\end{tcolorbox}
	\begin{tcolorbox}[colback=blue!5,colframe=blue!40!black,title=\textsc{Parameter estimation}]
		\textbf{\textit{c-5})} \textit{Alice} and \textit{Bob} sacrifice a relatively negligible part $\delta$ of the sifted-key bits in order to evaluate the \textit{quantum bit error rate} of the quantum communcation channel. Note that the protocol is considered to operate in an asymptotic regime ($N\rightarrow\infty$). As a result, only $(1-\delta)s\Tilde{\eta}N$ bits remain at the end of the parameter estimation procedure. For the asymptotic regime of operation, one assumes that $\delta\rightarrow 0$.
	\end{tcolorbox}
	\begin{tcolorbox}[colback=blue!5,colframe=blue!40!black,title=\textsc{Error correction}]
		\textbf{\textit{c-6})} $(1-\delta)(s\Tilde{\eta}N)fH(e)$ bits are announced by either \textit{Alice} or \textit{Bob} in order for error correction (information reconciliation) to be performed. The quantity $e$ denotes the quantum bit error rate. The quantity $f$ denotes the efficiency of the error correction algorithm (usually $f$ $=$ $1.1-1.2$). The function $H(\cdot)$ is the so-called \textit{Shannon entropy}. As a result, only $(1-\delta)s\Tilde{\eta}N$ $-$ $(1-\delta)(s\Tilde{\eta}N)fH(e)$ bits remain at the end of the error correction procedure.\\
	\end{tcolorbox}
	\begin{tcolorbox}[colback=blue!5,colframe=blue!40!black,title=\textsc{Privacy amplification}]
		\textbf{\textit{c-7})} $(1-\delta)s\Tilde{\eta}Q$ $-$ $(1-\delta)(s\Tilde{\eta}N)fH(e)$ $+$ $RN$ $-$ $1$ bits are announced by either \textit{Alice} or \textit{Bob} in order for privacy amplification to be performed. The quantity $R$, known as \textit{key} rate, is evaluated by the expression $R$ $=$ $(1-\delta)s(\xi - H(e) - fH(e))$ \cite{Gottesman}. The quantity $\xi$ denotes the confidential capacity of the protocol \cite{Bebrov2024}. Note that Toeplitz-based privacy amplification is considered. As a result, only $(1-\delta)s\Tilde{\eta}N$ $-$ $(1-\delta)(s\Tilde{\eta}N)fH(e)$ $-$ $[(1-\delta)s\Tilde{\eta}N$ $-$ $(1-\delta)(s\Tilde{\eta}N)fH(e)$ $+$ $RN$ $-$ $1]$ bits remain at the end of the privacy amplification procedure.
	\end{tcolorbox}
	
	Note that $s$ and $\sigma$ depend on the probability distribution of choosing a certain preparation/measurement basis $\{p,1-p\}$, where $p$ is the probability of using $Z$ (rectilinear) basis and $1-p$ is the probability of using $X$ (diagonal) basis.
	
	\section{Optimality}\label{opt-def}
	\textit{Note}: We consider a QKD operating in an asymptotic regime ($N\rightarrow\infty$, $\delta\rightarrow0$), \textit{\textit{i.e.}}, $N$ and $\delta$ are assumed to be fixed (see section \ref{bb84-steps} for reference concerning $N$ and $\delta$). That is why, $N$ and $\delta$ are not considered as protocol parameters.\\
	\indent Let first recall the quantity \textit{total efficiency} of a QKD. It is given by \cite{Bebrov}
	\begin{equation}
		\mathfrak{E} = \frac{R\cdot N}{N + M}
	\end{equation}
	where $N$ denotes the amount of qubits used in a QKD, $M$ denotes the amount of bits used in a QKD, and $R$ denotes the so-called secret key rate of a QKD. The quantity $R$ (single-photon secret key rate in this work) of a QKD model is given by \cite{Gottesman}
	\begin{equation}\label{skr-exp}
		R = \Tilde{\eta}\cdot s(\xi - H(e)-fH(e))
	\end{equation}
	where $f$ is the efficiency of the error correction algorithm, $s$ is sifting coefficient ($s \in (0,1]$), $H(\cdot)$ is Shannon entropy, $e$ is the error rate of the QKD system. Here $\Tilde{\eta}$ $=$ $\eta\cdot\eta_{\text{det}}$ represents the probability of detecting a quantum system at receiving side ($\eta$ $=$ $10^{\frac{-\alpha L}{10}}$---channel transmittance; $\eta_{\text{det}}$---detector's efficiency; $\alpha$---channel attenuation; $L$---length of the QKD system). The error rate $e$ is evaluated as follows \cite{Ma}
	\begin{equation}
		e = \frac{e_0p_{\text{dark}}+e_{\text{opt}}\Tilde{\eta}}{y_1}
	\end{equation}
	where $e_0$ is the background noise (usually $e_0 = 0.5$), $p_{\text{dark}}$ is the dark count rate of the detectors, $e_{\text{opt}}$ is the detection error rate (probability of a photon to be detected at the erroneous detector), and $y_1$ is the single-photon yield given by \cite{Ma}
	\begin{equation}
		y_1 = p_{\text{dark}} + \Tilde{\eta} - p_{\text{dark}}\Tilde{\eta}
	\end{equation}
	\indent The quantity $M$ is given by \cite{Bebrov}
	\begin{equation}\label{cbits-exp}
		M = \sum\limits_i M_i
	\end{equation}
	where $M_i$ is the amount of announced bits during the $i$th procedure (either sifting, or parameter estimation, or error correction, or privacy amplification)
	Based on protocol steps presented in section \ref{bb84-steps}, $M$ is a function of $s$ and $\sigma$.

	The above lines imply that the total efficiency $\mathfrak{E}$ is a function of $L$, $\alpha$, $\eta_{\text{det}}$, $p_{\text{dark}}$, $e_0$, $e_{\text{opt}}$, $f$, $s$, $\xi$. \\
	\indent Let now proceed to the \textit{optimality} of a QKD. The optimality is defined as
	\begin{equation}\label{opt-expr}
		\mathfrak{O} = \max \mathfrak{E}
	\end{equation}
	where the maximization is taken over all protocol parameters of a QKD model. For the sake of simplicity, we consider the case in which the protocol parameters $L$, $\alpha$, $\eta_{\text{det}}$, $p_{\text{dark}}$, $e_0$, $e_{\text{opt}}$, and $f$ are fixed. Fixing $L$, $\alpha$, $\eta_{\text{det}}$, $p_{\text{dark}}$, $e_0$, and $e_{\text{opt}}$, corresponds to utilizing a certain quantum communication channel, while fixing $f$ corresponds to utilizing a certain classical communication system for performing information reconciliation (error correction). Hence the total efficiency of a QKD remains a function of parameters $\xi$, $\sigma$, and $s$ ($\mathfrak{E}$ $=$ $f(\sigma,s,\xi)$):
	\begin{equation}
		\mathfrak{E} = \frac{R(s,\xi)\cdot N}{N + M(\sigma,s)}
	\end{equation}
	Note that increasing each variable ($s$, $\sigma$, and $\xi$) leads to increasing $\mathfrak{E}$, as verified by the above expression.\\
	\indent We now present an algorithm for obtaining the optimality of a QKD model.
	\begin{tcolorbox}[colback=magenta!5,colframe=magenta!40!black,title=\textsc{Algorithm (Method)---Determining optimality}]
		The algorithm aims to determine $\max\limits_{s,\sigma,\xi} \mathfrak{E}$. It consists of the following steps:
		\begin{enumerate}
			\item finding $\max(\xi)$ based on Ref. \cite{Bebrov2024}.
			\item utilizing an extremely biased choice of preparation/measurement bases in a QKD. This act leads to $s\rightarrow1$.
			\item implementing classical channel compression with $\sigma\rightarrow1$.
			\item evaluating $\mathfrak{E}$ for $\max(\xi)$, $\sigma\rightarrow1$, $s\rightarrow1$. 
		\end{enumerate}
	\end{tcolorbox}
	
	\section{Optimality Evaluation}\label{opt-eval}
	We evaluate the optimality of BB84-QKD protocol. Based on the protocol steps given in section \ref{bb84-steps}, the total efficiency of BB84-QKD has the following form
	\begin{equation}
		\mathfrak{E} = \frac{R\cdot N}{N + M} = \frac{R(s,\xi)\cdot N}{N + M(\sigma,s)}.
	\end{equation}
	Note that the parameters $L$, $\alpha$, $\eta_{\text{det}}$, $f$, $p_{\text{dark}}$, $e_{\text{opt}}$, and $e_0$ are fixed. Substituting for $R$ (see expression \eqref{skr-exp}) and $M$ in the last expression, one gets
	\begin{eqnarray}\label{eff-r-m}
		\mathfrak{E} = \frac{\Tilde{\eta}s(\xi-H(e)-fH(e))\cdot N}{N + [N+2(1-\sigma)\Tilde{\eta}N + \delta N + (1-\delta)(s\Tilde{\eta}N) + \Tilde{\eta}s(\xi-H(e)-fH(e))\cdot N - 1]}
	\end{eqnarray}
	Note that $M$ is obtained by adding up all the broadcasted (announced) bits in the protocol steps of section \ref{bb84-steps}, see \eqref{cbits-exp},
	\begin{eqnarray}
		M &&= N+2(1-\sigma)\Tilde{\eta}N + \delta N + (1-\delta)(s\Tilde{\eta}N)fH(e) + (1-\delta)(s\Tilde{\eta}N)  \nonumber \\
		&& - (1-\delta)(s\Tilde{\eta}N)fH(e) + \Tilde{\eta}s(\xi-H(e)-fH(e))\cdot N - 1 \nonumber \\
		&&= N+2(1-\sigma)\Tilde{\eta}N + \delta N + (1-\delta)(s\Tilde{\eta}N) + \Tilde{\eta}s(\xi-H(e)-fH(e))\cdot N - 1]
	\end{eqnarray}
	One can simplify expression \eqref{eff-r-m} as follows
	\begin{eqnarray}
		\mathfrak{E} &&= \frac{\Tilde{\eta}s(\xi-H(e)-fH(e))\cdot N}{N + [N+2(1-\sigma)\Tilde{\eta}N + \delta N + (1-\delta)(s\Tilde{\eta}N) + \Tilde{\eta}s(\xi-H(e)-fH(e))\cdot N - 1]} \nonumber \\
		&&= \frac{\Tilde{\eta}s(\xi-H(e)-fH(e))\cdot N}{N\cdot[2+2(1-\sigma)\Tilde{\eta} + \delta + (1-\delta)(s\Tilde{\eta}) + \Tilde{\eta}s(\xi-H(e)-fH(e)) - \frac{1}{N}]} \nonumber \\
		&&= \frac{\Tilde{\eta}s(\xi-H(e)-fH(e))}{2+2(1-\sigma)\Tilde{\eta} + \delta + (1-\delta)(s\Tilde{\eta}) + \Tilde{\eta}s(\xi-H(e)-fH(e)) - \frac{1}{N}}.
	\end{eqnarray}
	We consider an asymptotic regime of operation ($\delta\rightarrow0$, $N\rightarrow\infty$). The total efficiency of a BB84-QKD then becomes
	\begin{equation}
		\mathfrak{E} = \frac{\Tilde{\eta}s(\xi-H(e)-fH(e))}{2+2(1-\sigma)\Tilde{\eta} + s\Tilde{\eta} + \Tilde{\eta}s(\xi-H(e)-fH(e))}.
	\end{equation}
	Considering the optimal scenario of BB84-QKD ($\max(\xi)\overset{\text{BB84}}{=}1$, $\sigma\rightarrow1$, $s\rightarrow1$), one gets
	\begin{eqnarray}
		\mathfrak{O} &&= \mathfrak{E}_{\text{opt}} = \frac{\Tilde{\eta}(1-H(e)-fH(e))}{2+\Tilde{\eta} + \Tilde{\eta}(1-H(e)-fH(e))} = \frac{\Tilde{\eta}(1-H(e)-fH(e))}{\Tilde{\eta}[\frac{2}{\Tilde{\eta}} + 1 + (1-H(e)-fH(e))]} \nonumber \\
		&&= \frac{1-H(e)-fH(e)}{\frac{2}{\Tilde{\eta}} + 2 - H(e) - fH(e)}.
	\end{eqnarray}
	\indent In order to show the difference between optimal and standard BB84-QKD, we evaluate $\mathfrak{E}^{\text{BB84}}_{\text{opt}} (L)$ ($\mathfrak{O}^{\text{BB84}} (L)$) and $\mathfrak{E}^{\text{BB84}} (L)$, respectively. Note that $\mathfrak{E}^{\text{BB84}} (L)$ is evaluated given $s = 1/2$ and $\sigma = 0$, while $\mathfrak{O}^{\text{BB84}} (L)$ is evaluated given $s\rightarrow1$ and $\sigma\rightarrow1$. We have computed $\mathfrak{O}^{\text{BB84}} (L)$ and $\mathfrak{E}^{\text{BB84}} (L)$ in the case of a single-photon, asymptotic regime of operation and protocol parameters $\alpha$ $=$ $0.2$, $\eta_{\text{det}}$ $=$ $0.3$, $f$ $=$ $1$, $p_{\text{dark}}$ $=$ $10^{-8}$, $e_{\text{opt}}$ $=$ $0.03$, and $e_0$ $=$ $0.5$. The computation results are depicted in figure \ref{figure2}. Note that the red (top) plot in figure \ref{figure2} represents the efficiency limit of a BB84-QKD for the scenario considered herein (given the protocol parameters above).
	\begin{figure*}[!hbt]
		\centering
		\includegraphics[scale=0.6]{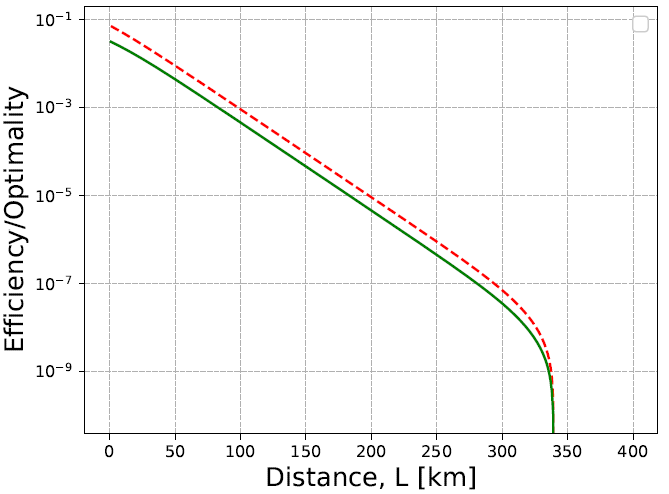}
		\caption{Total efficiency comparison between optimal BB84-QKD (presented in previous subsection) and standard BB84-QKD. Optimal BB84-QKD ($\mathfrak{O}^{\text{BB84}}(L)$)---dashed plot; Standard BB84-QKD ($\mathfrak{E}^{\text{BB84}}(L)$)---solid plot. \label{figure2}}
	\end{figure*}
	
	\section{Compression}\label{compression}
	In this section, a classical channel compression, being appropriate to obtain optimal QKD, is proposed.
	\subsection{Compression definition}
	As mentioned above, in order to make QKD optimal, we need to have efficient quantum and classical channels. A quantum channel becomes efficient when biased preparation and measurement bases are used \cite{Lo2005}. To make the public classical channel efficient, we utilize the concept of \textit{compression}. We define a compression that results in vanishing the classical channel. This compression is called \textit{channel squeezing}.
	\begin{definition}[Channel squeezing]
		It is a process given by the following combination of functions
		\begin{align}\label{compr}
			\mathfrak{compr}_k: & \{\text{0,1}\}^n \xrightarrow[preparing]{\mathcal{P}} \{\{\text{0,1}\}^k\}^m    \xrightarrow[coding]  {\mathcal{C}} \nonumber \\ & \{\{\text{0}\},\{\text{0,1}\}^{\text{2}},\{\text{0,1}\}^{\text{3}},...,\{\text{0,1}\}^r,\{\text{0,1}\}^r\}^m,
		\end{align}
		where $n$ is the number of bits to be announced over a public communication channel (this number of bits is related to the number of transferred qubits or number of setup uses), $m = n/k$ ($n=a\cdot k$, $a$ $\in$ $\mathbb{Z}_+$, $a$ $\rightarrow$ $\infty$ $\Rightarrow$ $n$ $\rightarrow$ $\infty$), and $r$ is related to $k$: $k = \text{log}_{\text{2}}(r+\text{1})$, $r$ = 0,1,..,2$^{k}-\text{1}$. \\
	\end{definition}
	The variable $k$ is called \textit{degree of compression}. Let $D$ = \{0,1\}, $P$ = \{\{0,1\}$^k$\}, $C$ = \{\{0\},\{0,1\}$^{\text{2}}$,\{0,1\}$^{\text{3}}$,...,\{0,1\}$^r$,\{0,1\}$^r$\} be sets of elements. The compression process $\mathfrak{compr}$, as shown in Eq. \eqref{compr}, is a composition of two functions $\mathcal{P}$ and $\mathcal{C}$ called \textit{preparing} and \textit{coding}, respectively. It identifies a transform of $n$ one-bit words (bits) into $m$ codewords with varying length $r$ ($r$ = 0,1,..,2$^{k}-\text{1}$), where $k$ is the parameter of the compression process. This parameter is related to the extent to which a binary sequence could be compressed, as will be shown later on.\\
	\indent The compression process $\mathfrak{compr}$ follows a typical compression algorithm: the messages to be compressed (elements of set $D$ or elements of set $P$) are ordered in a descending fashion\textemdash from the most probable to the least probable \cite{Huffman1952}. We then assign codewords (elements of set \textit{C}) to the ordered messages, as follows. We assign the first element of $C$ (\{0\}) to the most probable message, the second element of $C$ (\{0,1\}$^{\text{2}}$) to the second most probable message, and so forth. The codewords of length greater than 1 (\textit{\textit{i.e.}}, \{0,1\}$^{\text{2}}$, \{0,1\}$^{\text{3}}$, and so on) must meet the requirements of the conventional compression algorithms \cite{Huffman1952}: "\textit{...No message shall be coded in such a way that its code is a prefix of any other message, or that any of its prefixes are used elsewhere as a message code}". In other words, a codeword must not be assigned to two or more messages (it must be assigned to a single message) and a codeword must not play the role of a prefix of another codeword.  \\
	\indent As indicated in \cite{Huffman1952}, a compression process is characterized by the parameter \textit{average length} of the codewords (the elements of the above sets) that is given by
	\begin{equation}
		L_{av} = \sum_{i=\text{1}}^{|\mathcal{X}|} p(x_i)l(x_i),
	\end{equation}
	where $x_i$ identifies an element of a set $\mathcal{X}$, $|\mathcal{X}|$ identifies the cardinality of $\mathcal{X}$, $p(x_i)$ identifies the probability of an element $x_i$ to occur, and $l(x_i)$ identifies the length of the binary representation of an element $x_i$ (the length of the codeword assigned to $x_i$). Note that the product $L_{av} \cdot N$ yields the total amount of bits used to represent $N$ messages. \\
	\indent The compression of a sequence of messages could be quantitatively evaluated by the following expression
	\begin{equation}\label{sigma}
		\sigma = \Big(\text{1} - \frac{L_{av,C}\cdot m}{L_{av}\cdot n}\Big)\cdot \text{100}\; [\%],
	\end{equation}
	where $L_{av,C}$ is the average length obtained after performing $\mathcal{C}$ (average length of set $C$) and $L_{av}$ is the average length of set $D$ (initial set). \\
	\indent For the sake of clarity, we present an example for the channel squeezing in Appendix~\ref{app2}. 
	\subsection{Compression results}\label{results}
	
	\indent What is the possible compression coefficient $\sigma$ that could be achieved for the classical channel of a QKD? To answer this question, we evaluate the function $\sigma$($k$), which is displayed in figure \ref{figure}. We use $n$ = 10$^{\text{30}}$ to obtain the values of this relation. We assume that this value of $n$ is comparable to $n\rightarrow\infty$. The figure clearly shows that $\sigma\rightarrow\text{1}$ (100\%) when $k\rightarrow\infty$. This implies that we could completely compress the public classical channel of a QKD when $k\rightarrow\infty$.
	This asymptotic behavior of $\sigma$ is justified in the proof of the following theorem. 
	
	\begin{figure}[b]
		\centering
		\includegraphics[width=0.5\textwidth]{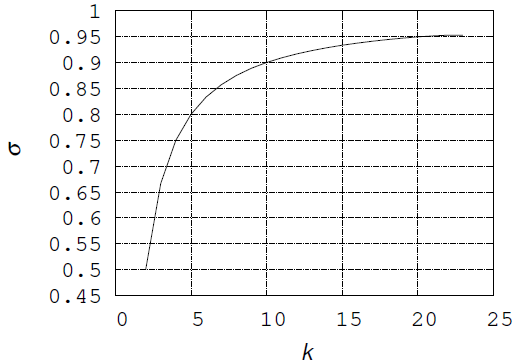}
		\caption{Relation between the compression coefficient $\sigma$ and the degree of compression $k$. The plot is obtained for $n$ = 10$^{\text{30}}$. Due to lack of computation resources, we calculate $\sigma$ only for $k$ = $\text{2,3,...,24}$.\label{figure}}
	\end{figure}

	Based on the above results for the channel squeezing, the following theorem is presented. In the proof of the theorem, we rigorously verify the behavior of $\sigma$. 
	
	\begin{theorem}
		A quantum key distribution protocol (considering only state preparation, state measurement, and sifting procedures) could asymptotically approach the optimal quantum key distribution provided that both \textbf{biased preparation/measurement bases} \text{($p\rightarrow \text{1}$)} and \textbf{compression process $\mathfrak{compr}_k$} are feasible (and used), and the number of transferred qubits tends to infinity \text{($n\rightarrow\infty$)}. 
	\end{theorem}
	\begin{proof}
		We consider a general QKD. Let $(p,\text{1}-p)$ denote the probability distribution of a random variable $Y$ representing the preparation/measurement bases ($Y$ = \{$Z$,$X$\}, $Z$ and $X$ are the well-known polarization bases \cite{Bennett1984} or phase bases \cite{Lucamarini2018}).  As shown in \cite{Lo2005}, without compromising the security of the QKD model, $p$ could go to unity given that $n$ $\rightarrow$ $\infty$ ($n$\textemdash number of transferred quantum systems in a QKD protocol or number of uses of a QKD system). Since $p\rightarrow \text{1}$ and the number of quantum systems used for checking eavesdropping is fixed, we could assume that almost one preparation/measurement basis (that with probability $p$) is used
		for establishing key bits ($p^{\text{2}}$ is the probability of establishing a key bit with this encoding basis). The other basis (that with probability $\text{1}-p$) is also used for establishing key bits \cite{Lo2005}, however, it contributes considerably less to the establishing process ($(\text{1}-p)^{\text{2}}$ is the probability of establishing a key bit with this basis).
		Therefore, we could assume that one of the bases is used for establishing key bits. Thus, the probability of establishing key bits (in the asymptotic case) is $p^{\text{2}}$. Given that $p\rightarrow \text{1}$, then $p^{\text{2}}$ also tends to unity. In this regard, the number of established key bits, say $K$, asymptotically tends to $n$ ($K$ $\rightarrow$ $n$). This implies that the efficiency of the quantum channel tends to unity (or 100\%) for the considered case ($p\rightarrow \text{1}$, $n\rightarrow\infty$). This analysis is based on the (yet presented in the) work of Ref. \cite{Lo2005}.\\
		\indent In the case of ($p\rightarrow \text{1}$, $n\rightarrow\infty$), it is possible that $k\rightarrow\infty$ [\textit{Note}: $k\rightarrow\infty$ is possible if and only if  $n\rightarrow\infty$]. For $k\rightarrow\infty$, the compression coefficient $\sigma$ of $\mathfrak{compr}_k$ asymptotically approaches unity (or 100\%). This is proved in the following way. For $p\rightarrow\text{1}$, $L_{av,C}\rightarrow\text{1}$. This could be readily verified by using the example presented in Table \ref{examp2} and by the expression
		\begin{equation}
			L_{av,C} = \sum_{i=\text{0}}^{{\text{2}}^k-\text{1}} l(x_i)p(x_i),
		\end{equation}
		where 
		\begin{equation}
			l(x_i) = \begin{cases}
				i+\text{1}, \; i < {\text{2}}^k-\text{1} \\
				i, \; i = {\text{2}}^k-\text{1}
			\end{cases}
		\end{equation}
		and
		\begin{equation}
			p(x_i) = \begin{cases}
				p^{\text{log}_{\text{2}}k},  \; i=\text{0}  \\
				p^{\text{log}_{\text{2}}k-\text{1}}(\text{1}-p), \; i=\text{1} \\
				p^{\text{log}_{\text{2}}k-\gamma_i}(\text{1}-p)^{\gamma_i}, \; i > \text{1}
			\end{cases}
		\end{equation}
		where $x_i$ is an element of set $P$ (\textit{e.g.}, $P$ = $\{\text{0,1}\}^{\text{2}}$ $\Rightarrow$ $x_{\text{0}}$ = 00, $x_{\text{1}}$ = 01, $x_{\text{2}}$ = 10, and $x_{\text{3}}$ = 11). The parameter $\gamma_i$ is given by $\text{1}+\gamma_{i-\text{2}^{\left \lfloor{\text{log}_\text{2}i}\right \rfloor}}$ and identifies the number of 1s presented in an element $x_i$ [\textit{Note:} The above expression of $\gamma_i$ holds for $i>\text{1}$, $\gamma_{\text{0}}$ = 0 and $\gamma_{\text{1}}$ = 1 by default].  Also, $L_{av}$ is always equal to unity (by definition). This leads to the following form of $\sigma$ (see for reference Eq. \eqref{sigma})
		\begin{equation}
			\sigma = \Big(\text{1} - \frac{m}{n}\Big)\cdot \text{100}
		\end{equation}
		Substituting $m$ by $n/k$ (according to its definition given above), we yield
		\begin{equation}
			\sigma = \Big(\text{1} - \frac{\frac{n}{k}}{n}\Big)\cdot \text{100},
		\end{equation}
		which takes the final form
		\begin{equation}
			\sigma = \Big(\text{1} - \frac{\text{1}}{k}\Big)\cdot \text{100}.
		\end{equation}
		As stated above, when $k\rightarrow\infty$, $\sigma\rightarrow\text{1}$ (100\%).\\
		\indent This result implies that the classical channel of the considered QKD could be completely compressed, \textit{\textit{i.e.}}, it could be neglected. \\
		\indent As a result, the considered QKD protocol (operating at $p\rightarrow \text{1}$, $n\rightarrow\infty$, $k\rightarrow\infty$) obtains an asymptotically ideal quantum channel and asymptotically irrelevant classical channel (a channel to be neglected or ideal classical channel). In this way, we prove the theorem. 
	\end{proof}

	\section{Optimal QKD protocols}\label{protocols}
	Based on optimality definition and evaluation (sections \ref{opt-def} and \ref{opt-eval}), we present two \textit{optimal} QKD protocols (protocols in which both extremely biased preparation/measurement bases \cite{Lo2005} and channel squeezing are involved) in the following lines. The two protocols are: \textbf{(I)} a \textit{point-to-point} (\textit{BB84-like}) \cite{Bennett1984,Lo2005} \textit{optimal QKD protocol}; \textbf{(II)} a \textit{point-relay-point} (twin-field like, \textit{Lucamarini2018-like}) \cite{Lucamarini2018,Curty2019} \textit{optimal QKD protocol}. We use channel squeezing of type $\mathfrak{compr}_k$, where $k\rightarrow\infty$. \\
	\indent The protocol of type \textbf{(I)} (operating in the \textit{single-photon regime}) is characterized by the following steps: \\
	\textbf{Step 1.} Alice prepares a sequence of $N$ qubits ($N\rightarrow\infty$). Each qubit $q_n$ resides in the state $\hat{H}(b_n)\ket{q_n}$, where $n$ $=$ $1$,$2$,$...$,$N$ ($n$ $\in$ $\mathcal{N}$, $\mathcal{N} \subset \mathbb{Z}_+$, $N$ = $|\mathcal{N}|$ is cardinality of $\mathcal{N}$). Note that $q$ $\in$ \{0,1\}, \textit{i.e.}, $\ket{q_n}$ $\in$ \{$\ket{\text{0}_n}$,$\ket{\text{1}_n}$\}. The symbol $q$ denotes a realization of a random variable $Q$ ($P_Q$ = ($p_{Q}$,$\text{1}-p_{Q}$); $p_Q$ = $p(q=\text{0})$, $\text{1}-p_Q$ = $p(q=\text{1})$; $p_Q$ = 0.5). The operator $\hat{H}(b)$ is a \textit{conditional Hadamard operator}. It is defined as follows
	\begin{equation}
		\hat{H}(b) = \begin{cases}
			\hat{I} \; \text{if $b = 0$;}\\
			\hat{H} \; \text{if $b = \text{1}$.}
		\end{cases}
	\end{equation}
	where $\hat{I}$ is the identity operator and $\hat{H}$ is the Hadamard operator. Applying $\hat{I}$ corresponds to a qubit $q_n$ prepared in a $Z$ polarization basis. Applying $\hat{H}$ corresponds to a qubit $q_n$ prepared in a $X$ polarization basis. The random variable $B = b$ identifies the so-called \textit{basis} (\textit{e.g.}, polarization basis) in terms of which a qubit is prepared. It is characterized with probability distribution $P_B$ = ($p_B$,$\text{1}-p_B$) [\textit{Note}: $p_B$ = $p(b=\text{0})$; $\text{1}-p_B$ = $p(b=\text{1})$]. In the presented protocol, $P_B$ is not uniform ($p_B$ $\neq$ 0.5, $p_B\rightarrow1$), as in the protocol of Ref. \cite{Lo2005}. Alice sends out the $N$ qubits to Bob over a quantum communication channel. \\
	\textbf{Step 2.} Bob receives the $N$ qubits (in this representation of the protocol, we assume that there is no loss of qubits during their transfer). He performs a measurement $\hat{M}_n$ to each qubit $q_n$. The measurement operator is given by its spectral decomposition 
	\begin{equation}
		\hat{M}_n = \sum_q \lambda_{q,n}\hat{H}(b'_n)\ket{q_n}\bra{q_n}\hat{H}^{\dagger}(b'_n),
	\end{equation}
	where $\lambda_{q,n}$ $\in$ \{$\lambda_{\text{0},n}=\text{1}$,$\lambda_{\text{1},n}=-\text{1}$\} are the eigenvalues of the measurement operator $\hat{M}_n$ and $b'$ is a realization of a random variable $B'$ ($P_B'$ = ($p_{B'}$,$\text{1}-p_{B'}$); $p_{B'}$ $\neq$ 0.5). The operator $\hat{H}(b'_n)$ is defined similarly to  $\hat{H}(b_n)$:
	\begin{equation}
		\hat{H}(b') = \begin{cases}
			\hat{I} \; \text{if $b' = \text{0}$;}\\
			\hat{H} \; \text{if $b' = \text{1}$.}
		\end{cases}
	\end{equation}
	Based on the results of the measurements $\hat{M}_n$ (the results are eigenvalues $\lambda_{q,n}$), Bob forms the so-called \textit{raw key} that is comprised of bits $k_{B,n}$, where $k_{B,n}$ = 0 if $\lambda_{q,n}=\text{1}$ and $k_{B,n}$ = 1 if $\lambda_{q,n}=-\text{1}$. Point out that Alice forms her \textit{raw key} bits $k_{A,n}$ based on the variable $Q=q$: $k_{A,n}$ = $q$ ($k_{A,n}$ = 0 if $q=\text{0}$ and $k_{A,n}$ = 1 if $q=\text{1}$). \\
	\textbf{Step 3.} Bob announces $\mathfrak{compr}_k$($b'_n$). Alice decompresses $\mathfrak{compr}_k$($b'_n$) [$\mathfrak{compr}_k^{-\text{1}}$($\mathfrak{compr}_k$($b'_n$))] in order to obtain the original values of $b'_n$. Alice compares her $b_n$ to Bob's $b'_n$. Based on this comparison, Alice forms a sequence of $N$ bits $d_n$, where each $d_n$ is determined by 
	\begin{equation}
		d_n = \begin{cases}
			\text{0} \; \text{if $b_n=b'_n$;}\\
			\text{1}  \; \text{otherwise.}
		\end{cases}
	\end{equation}
	The bits $d_n$ are publicly announced by Alice. Actually, Alice announces $\mathfrak{compr}_k$($d_n$) (a compressed sequence of $d_n$). In order to obtain the original $d_n$, Bob decompresses $\mathfrak{compr}_k$($d_n$) [$\mathfrak{compr}_k^{-\text{1}}$($\mathfrak{compr}$($d_n$))]. These bits are used to correlate the raw keys of Alice and Bob. This procedure (correlation) is known as \textit{sifting}. In the sifting, Alice and Bob discard those bits of their initial (raw) keys $k_{A,n}$ and $k_{B,n}$, for which $b_n \neq b'_n$ ($d_n=\text{1}$). As a result, Alice and Bob yield the so-called \textit{sifted keys} $k_{A,f}$ and $k_{B,f}$, respectively. Note that $f \in \mathcal{F}$, $\mathcal{F}$  $\subset$ $\mathcal{N}$, $F$ = $|\mathcal{F}|$ is cardinality of $\mathcal{F}$. Based on $d_n$, Alice (Bob) reduces her (his) sequence of key bits $k_{A,n}$ to a sequence of $k_{A,f}$.  \\
	\textbf{Step 4.} The \textit{parameter estimation} procedure begins. Alice (Bob) divides her (his) sifted key bits $k_{A,f}$ ($k_{B,f}$) into $k_{A,v}$ ($k_{B,v}$) and $k_{A,w}$ ($k_{B,w}$). The bits $k_{A,v}$ are those obtained when $X$ basis is used for both preparation and measurement. The bits $k_{A,w}$ are those obtained when $Z$ basis is used for both preparation and measurement. Note that $v \in \mathcal{V}$, $\mathcal{V}$ $\subset$ $\mathcal{F}$, $V$ = $|\mathcal{V}|$ is cardinality of $\mathcal{V}$; $w \in \mathcal{W}$, $\mathcal{W}$ = $\mathcal{V}$$^c$ $\subset$ $\mathcal{F}$, $W$ = $|\mathcal{W}|$ is cardinality of $\mathcal{W}$. Alice picks $k_{A,v'}$ bits out of $k_{A,v}$, and $k_{A,w'}$ bits out of $k_{A,w}$. These bits will be used for evaluating two error rates\textemdash error rate of $X$ basis and error rate of $Z$ basis \cite{Lo2005}. Here $v' \in \mathcal{V'}$, $\mathcal{V'}$ $\subset$ $\mathcal{V}$, $V'$ = $|\mathcal{V'}|$ is cardinality of $\mathcal{V'}$; $w' \in \mathcal{W'}$, $\mathcal{W'}$ $\subset$ $\mathcal{W}$, $W'$ = $|\mathcal{W'}|$ is cardinality of $\mathcal{W'}$. Alice announces $k_{A,v'}$ and $k_{A,w'}$. Bob determines the following two \textit{quantum bit error rates} $QBER_X$ and $QBER_Z$:
	\begin{equation}
		QBER_X = \frac{\sum_{v'} g_{v'}(k_{A,v'},k_{B,v'})}{V'},
	\end{equation}
	\begin{equation}
		QBER_Z = \frac{\sum_{w'} g_{w'}(k_{A,w'},k_{B,w'})}{W'},
	\end{equation}
	where the function $g_x(k_{A,x},k_{B,x})$ is given by
	\begin{equation}
		g_x(k_{A,x},k_{B,x}) = \begin{cases}
			\text{0} \; \text{if $k_{A,x}=k_{B,x}$;}\\
			\text{1} \; \text{otherwise.}
		\end{cases}
	\end{equation}
	$x$ being a dummy variable. Alice and Bob terminate the current session of the protocol if both $QBER_X$ and $QBER_Z$ exceed a preliminarily determined threshold $\epsilon$ ($QBER_{X(Z)} > \epsilon$) \cite{Lo2005}. Otherwise, Alice and Bob proceed forward [\textit{Note:} Bob sends to Alice a binary message: 0 corresponds to \textit{terminating the protocol} and 1 corresponds to \textit{proceeding forward}]. After the parameter estimation procedure (this step), Alice (Bob) possesses a key $k_{A,v''}\textbf{\&}k_{A,w''}$ ($k_{A,v''}\textbf{\&}k_{A,w''}$), where $\textbf{\&}$ identifies concatenation of bits. Here $v'' \in \mathcal{V''}$, $\mathcal{V''}$ = $\mathcal{V'}$$^c$ $\subset$ $\mathcal{V}$, $V''$ = $|\mathcal{V''}|$ is cardinality of $\mathcal{V''}$; $w'' \in \mathcal{W''}$, $\mathcal{W''}$ = $\mathcal{W'}$$^c$ $\subset$ $\mathcal{W}$, $W''$ = $|\mathcal{W''}|$ is cardinality of $\mathcal{W''}$.\\
	\textbf{Step 5.} Key reconciliation is performed on $k_{A,v''}\textbf{\&}k_{A,w''}$ and $k_{A,v''}\textbf{\&}k_{A,w''}$ \cite{Mehic2020,Brassard1994}. \\
	\textbf{Step 6.} Privacy amplification is performed \cite{Bennett1995,Yuan2018}. \\
	\indent The cardinalities of the subsets in the above protocol have the following magnitudes and are related to each other as follows: $N\rightarrow\infty$; $F$ = $\left \lfloor{[p_B^{\text{2}}+(\text{1}-p_B)^{\text{2}}]N}\right \rfloor$; $V$ = $\left \lfloor{p_B^{\text{2}}N}\right \rfloor$; $W$ = $\left \lfloor{(\text{1}-p_B)^{\text{2}}N}\right \rfloor$; $V' = \varepsilon V$; $W' =  \lambda W$; $V''$ = $V - V'$; $W''$ = $W - W'$. Here $\varepsilon$ and $\lambda$ are considerably small decimal numbers. \\
	\indent The protocol of type \textbf{(II)} (operating in the \textit{weak-coherent-pulse regime}) is characterized by the following steps (adopted from \cite{Curty2019}): \\
	\textbf{Step 1.} Alice (Bob) prepares $N$ weak coherent pulses (WCPs) $w_{n,A(B)}$ [\textit{Note}: we use the notation from the previously introduced protocol: $n\in\mathcal{N}$, $\mathcal{N}\subset\mathbb{Z}_+$, $N$ = $|\mathcal{N}|$ is the cardinality of $\mathcal{N}$, $N\rightarrow\infty$]. She (he) chooses a preparation basis $X$ ($Z$) with probability $p(X)$ ($p(Z) = \text{1} - p(X)$) for each $w_{n,A(B)}$. The probability $p(X)$ tends to unity ($p(X)\rightarrow1$). If $X$ is chosen, Alice (Bob) sets a $w_{v,A(B)}$ in a coherent state $\ket{\alpha_v}_{A(B)}$ when she (he) wishes to establish a key bit $b_A = \text{0}$ ($b_B = \text{0}$). Alice (Bob) sets a $w_{v,A(B)}$ in a coherent state $\ket{-\alpha_v}_{A(B)}$ when she (he) wishes to establish a key bit $b_A = \text{1}$ ($b_B = \text{1}$). Here $v\in\mathcal{V}$, $\mathcal{V}\subset\mathcal{N}$, $V$ = $|\mathcal{V}|$ is the cardinality of $\mathcal{V}$; the subset $\mathcal{V}$ identifies the weak coherent pulses prepared in $X$ basis.  The key bits $b_A$ and $b_B$ are randomly chosen (by Alice and Bob, respectively) and independently from each other. If $Z$ is chosen, Alice (Bob) sets a $w_{u,A(B)}$ in a phase-randomizing state $\ket{\beta_u}_{A(B)}$. Here $u\in\mathcal{U}$, $\mathcal{U}=\mathcal{V}$$^c$ $\subset\mathcal{N}$, $U$ = $|\mathcal{U}|$ is the cardinality of $\mathcal{U}$; the subset $\mathcal{U}$ identifies the weak coherent pulses prepared in $Z$ basis.  The amplitude $\beta$ of the phase-randomizing state is picked from the set $S=\{\beta_i\}_i$, where $\beta_i$ are real non-negative numbers ($\beta_i \geq \text{0}$). Alice (Bob) picks an amplitude $\beta_i$ for each $w_{u,A(B)}$ in accordance with a probability distribution $p_{\beta,A(B)}$.   \\
	\textbf{Step 2.} Alice and Bob send their coherent states $w_{n,A(B)}$  over an optical channel with transmittance $\sqrt{\eta}$ to Charlie (a relay node). They synchronize the transfer of their states.         \\
	\textbf{Step 3.} Charlie interferes the coherent states $w_{n,A}$ and $w_{n,B}$ at a balanced beamsplitter. A threshold detector is placed at each output of the beamsplitter: detector $D_c$ at output $c$ and detector $D_d$ at output $d$ ($c$\textemdash output that corresponds to constructive interference; $d$\textemdash output that corresponds to destructive interference).        \\
	\textbf{Step 4.} For each interference between $w_{n,A}$ and $w_{n,B}$, Charlie announces measurement outcome $c_n$ ($d_n$) corresponding to detector $D_c$ ($D_d$), where $c_n = \text{0}$ and $c_n = \text{1}$ ($d_n = \text{0}$ and $d_n = \text{1}$) indicate a non-click event and a click event, respectively.          \\ 
	\textbf{Step 5.} Both Alice and Bob announce their preparation bases $h_{A(B)}$. For this purpose, they use $\mathfrak{compr}_k$: $\mathfrak{compr}_k(h_{A(B)})$. The value of $k$ should be as greater as possible. In order to learn the preparation bases of Bob (Alice) from the announcement $\mathfrak{compr}_k(h_{A(B)})$, Alice (Bob) decompresses $\mathfrak{compr}_k(h_{B})$ ($\mathfrak{compr}_k(h_{A})$). Based on this announcement, Alice and Bob sift their sequences of coherent states $w_{n,A}$ and $w_{n,B}$. They discard those states for which their bases do not match. After this procedure, Alice and Bob are left with sequences of coherent states $w_{f,A}$ and $w_{f,B}$, respectively. Here $f\in\mathcal{F}$, $\mathcal{F}\subset\mathcal{N}$, $F$ = $|\mathcal{F}|$ $<$ $N$ is the cardinality of $\mathcal{F}$, and $\mathcal{U'}$,$\mathcal{V'}$ $\supset$ $\mathcal{F}$. The subsets $\mathcal{U'}$ and $\mathcal{V'}$ ($\mathcal{U'}=\mathcal{V'}$$^c$) identify sifted $Z$-basis states (phase-randomizing states) and sifted $X$-basis coherent states, respectively. Based on the interference of $w_{u',A}$ and $w_{u',B}$ (correspondingly, the announced measurement outcomes for these pulses), Alice and Bob form the so-called \textit{sifted key}. Note that when Charlie reports $c_{u'} = \text{1}$ and $d_{u'} = \text{0}$ ($c_{u'} = \text{0}$ and $d_{u'} = \text{1}$), the $u'$th bits of their sifted keys are related as follows: $b_{u',A}$ = $b_{u',B}$ ($b_{u',A}$ = $b_{u',B}\oplus\text{1}$).\\
	\textbf{Step 6.} Alice and Bob perform the parameter estimation procedure, which involves evaluating the \textit{QBER} and checking the presence of an eavesdropper by using the phase-randomizing pulses $w_{v',A}$ and $w_{v',B}$ of $Z$ basis (the decoy-state method is implemented for this purpose).       \\
	\textbf{Step 7.} Alice and Bob perform key reconciliation. \\
	\textbf{Step 8.} Alice and Bob perform privacy amplification.

	\section{Summary}\label{summary}
	In summary, this work introduces a way to determine and obtain the optimality of a quantum key distribution. The quantity \textit{optimality} ($\mathfrak{O}$) is defined. It is given by the maximization of the total efficiency of a QKD over its protocol parameters, see expression \eqref{opt-expr}. It is verified that an optimal QKD is realizable when using an extremely biased choice of preparation and measurement bases\textemdash capacity-reaching quantum channel and a completely compressed classical channel are attainable. The process of compression is used to make the public classical channel of a QKD system efficient\textemdash to almost (asymptotically) remove it from the QKD picture.
	A comparison between optimal BB84-QKD (function $\mathfrak{O}^{\text{BB84}}(L)$) and a standard BB84-QKD (function $\mathfrak{E}^{\text{BB84}}(L)$) is performed, see figure \ref{figure2}. As expected, it is verified that the optimal implementation displays higher efficiency (optimality) values.
	Moreover, two optimal QKD protocols are presented: a BB84-like optimal QKD and a Lucamarini2018-like optimal QKD. In this way, we demonstrate how optimal QKD  could be achieved (both extremely biased preparation/measurement bases \cite{Lo2005} and complete classical channel compression could be applied to a QKD).
	
	\section*{Acknowledgment}
	This work is supported by Technical University of Varna under Scientific Project Programme, Grant No. $\text{H}\Pi$6/2025. This research is also incorporated into the framework of Scientific Programme "Increasing National Scientific Capacity in the Field of Quantum Information", Bulgarian Ministry of Education and Science.

	\ifCLASSOPTIONcaptionsoff
	\newpage
	\fi
	
	\appendices
		\section{Example of channel squeezing}
		\label{app2}
		\indent Consider announcement of measurement bases ($Z$ and $X$ bases) in a QKD process. Suppose $p(Z)\rightarrow\text{1}$, correspondingly $p(X)\rightarrow\text{0}$ ($p(X) = \text{1} - p(Z)$). Suppose further that the use of measurement basis $Z$ is announced by binary 0, whereas the use of measurement basis $X$ is announced by binary 1. Consider now that $\mathfrak{compr}_k=\mathfrak{compr}_{\text{2}}$ (composition of functions $\mathcal{P}$ and $\mathcal{C}$) is performed on the announced sequence of 0s and 1s (sequence of measurement bases). This implies that the elements of the set $P$ (result of function $\mathcal{P}$) are of length 2\textemdash \{0,1\}$^k$ = \{0,1\}$^{\text{2}}$, as shown in Eq. \eqref{compr}. Table \ref{examp1} depicts the result of function $\mathcal{P}$ (set $P$), the first step of $\mathfrak{compr}_{\text{2}}$. The table also involves the probabilities of occurrence of the elements of set $P$. The elements of set $P$ are ordered in a descending fashion, as required by the compression algorithm itself (as mentioned above). Now we perform function $\mathcal{C}$ on the elements of set $P$. According to the compression algorithm mentioned above, we assign a corresponding element of set $C$, as shown in Table \ref{examp2}, to the corresponding ordered element of set $P$. In this way, we complete the example.    \\
		\begin{table}[!hbt]
			\centering
			\caption{\label{examp1}%
				Performing function $\mathcal{P}$ (of $\mathfrak{compr}_{\text{2}}$) on a sequence of bits (sequence of announced measurement bases). The elements of set $P$ (result of $\mathcal{P}$) are ordered in a descending fashion in terms of their probabilities of occurrence. Suppose that $p(Z)$ = 0.999 (corresponding to $p(Z)\rightarrow$ 1) and $p(X)$ = 0.001 (corresponding to $p(X)\rightarrow$ 0).}
			
			\begin{tabular}{cl}
				
				\toprule
				Elements of set $P$ & Probabilities \\
				\toprule
				00 & $p(Z)^{\text{2}}$ = 0.998001 \\
				01 & $p(Z)p(X)$ = 0.000999 \\
				10 & $p(X)p(Z)$ = 0.000999 \\
				11 & $p(X)^{\text{2}}$ = 10$^{\text{-6}}$ \\
				\bottomrule
			\end{tabular}
			
		\end{table}
		
		\begin{table}[!hbt]
			\centering
			\caption{\label{examp2}%
				Performing function $\mathcal{C}$ (of $\mathfrak{compr}_{\text{2}}$) on the ordered elements of set $P$. This table displays the final result of $\mathfrak{compr}_{\text{2}}$\textemdash the elements of $C$.}
			
			\begin{tabular}{cll}
				
				\toprule
				Elements of set $P$ & Probabilities & Elements of set $C$ \\
				\toprule
				00 & $p(Z)^{\text{2}}$ = 0.998001 & \{0\} = 0 \\
				01 & $p(Z)p(X)$ = 0.000999 & \{0,1\}$^{\text{2}}$ = 10 \\
				10 & $p(X)p(Z)$ = 0.000999 & \{0,1\}$^{\text{3}}$  = 110 \\
				11 & $p(X)^{\text{2}}$ = 10$^{\text{-6}}$ & \{0,1\}$^{\text{3}}$  = 111 \\
				\bottomrule
			\end{tabular}
			
		\end{table}

\end{document}